\documentclass[a4paper,twoside,reqno,11pt]{article}
\usepackage{amssymb,amsmath,amsthm}
\usepackage{graphicx} 
\usepackage{times}
\usepackage{color}
\usepackage{natbib}
\leftmargini=\baselineskip
\newtheoremstyle{localthm}
	{7pt} 
	{7pt} 
	{\sl} 
	{} 
	{\bf} 
	{{\rm.}} 
	{.7em} 
	{} 

\theoremstyle{localthm}
\newtheorem{Theorem}{Theorem}

\newtheorem{Lemma}[Theorem]{Lemma}

\newtheoremstyle{localrem}
	{5pt} 
	{5pt} 
	{\rm} 
	{} 
	{\bf} 
	{{\rm.}} 
	{.7em} 
	{} 

\theoremstyle{localrem}

\def\rss{\mathrm{rss}}
\def\RSS{\mathrm{RSS}}

\newcommand{\R}{\mathbb{R}}

\newcommand{\NN}{\mathcal{N}}

\def\rss{\mathrm{rss}}
\def\ss{\mathrm{ss}}
\def\RSS{\mathrm{RSS}}

\def\V{\mathbb{V}}

\def\bs{\boldsymbol}

\def\b{\bs{b}}
\def\bmu{\bs{\mu}}

\def\x{\bs{x}}
\def\X{\bs{X}}
\def\y{\bs{y}}
\def\Y{\bs{Y}}

\def\Z{\bs{Z}}

\def\hat{\widehat}
\def\tilde{\widetilde}

\begin{document}
\addtolength{\baselineskip}{0.4\baselineskip}
\title{Covariate Selection Based on a Model-free Approach to Linear Regression with Exact Probabilities 
}
\author{Laurie Davies and Lutz D\"umbgen\\
University of Duisburg-Essen, University of Bern}
\date{\today}
\maketitle

\begin{abstract}
In this paper we give a completely new approach to the problem of covariate selection in linear regression.  A covariate or a set of covariates is included only if it is  better in the sense of least squares than the same number of Gaussian covariates consisting of i.i.d. $N(0,1)$ random variables. The Gaussian P-value is defined as the probability that the Gaussian covariates are better. It is given in terms of the Beta distribution, it is exact and it holds for all data making it model-free free.  The covariate selection procedures require only a cut-off value $\alpha$ for the Gaussian P-value: the default value in this paper  is $\alpha=0.01$.  The resulting procedures are very simple, very fast, do not overfit and require only least squares. In particular there is no regularization parameter, no data splitting, no use of simulations, no shrinkage and no post selection inference is required. The paper includes the results of simulations, applications to real data sets and theorems on the asymptotic behaviour under the standard linear model. Here the step-wise procedure performs overwhelmingly better than any other procedure we are aware of. An R-package {\it gausscov} is available.
\end{abstract}




\section{Introduction} \label{sec:intro}

 Most statistical problems can be interpreted as ones of distinguishing a signal, here a relevant covariate, from noise. In this paper this is accomplished in a direct manner by comparing each covariate $\x_i$ with Gaussian i.i.d. noise generated by the statistician. The comparison is based on the Gaussian P-value $P_G(\x_i)$ which is defined as the probability that Gaussian noise is better than the covariate as measured by the reduction in the sum of squared residuals. 

More precisely consider a dependent variable $\y$ of size $n$ and $q$ covariates $\x_i, i=1,\ldots,q$. Regress $\y$ on a subset ${\mathcal S}$ of size $k<n-1$ with sum of squared residuals by $\rss_k$. Now include a Gaussian covariate $\Z_1$ consisting of $n$ i.i.d. $N(0,1)$ random variables and   regress $\y$ on $\mathcal{S} \cup \Z_1$ with sum of squared residuals $\RSS$. Then it follows from Theorem~\ref{the:RSS} below
\begin{equation} \label{equ:RSS_dist}
\RSS/\rss_k\sim \text{Beta}((n-k-1)/2,1/2)
\end{equation}
where $\text{Beta}(a,b)$ denotes a Beta random variable with parameters $(a,b)$. The whole of this paper derives from (\ref{equ:RSS_dist}) and a slightly more general result when $\y$ is regressed on $\mathcal{S}\cup\{\Z_1,\ldots,\Z_k\}$. Indeed in a sense (\ref{equ:RSS_dist}) {\it is} the paper. The result is rather surprising for two reasons. Firstly, it is model free as it depends only on $\rss_k$ and, secondly, the distribution can be stated exactly. In particular   (\ref{equ:RSS_dist}) remains valid no matter how the subset $\mathcal{S}$ was chosen.

We can use (\ref{equ:RSS_dist}) to define a P-value of a covariates $\x_i$ in the simplest situation where $\mathcal{S}$ is the set of all covariates and $q<n-1$ as follows.  Regress $\y$ on all $q$ covariates with sum of squared residuals $\rss_q$. Replace $\x_i$ by a Gaussian covariate $\Z_i$ and regress $\y$ on $(\mathcal{S}\setminus\{\x_i\})\cup\{\Z_i\}$ to give a sum of squared residuals $\RSS_i$. The Gaussian P-value of $\x_i$ is defined by
 \begin{equation} \label{equ:P_val_1}
P_G(\x_i)=\bs{P}(\RSS_i \le \rss_q)= \text{Beta}_{(n-q)/2,1/2}(\rss_q/\rss_{q,-i})
\end{equation}
where $\rss_{q,-i}$ is the sum of squared residuals based on all  $\mathcal{S}\setminus\{\x_i\}$ and $\text{Beta}_{a,b}$ denotes the distribution function of the Beta distribution with parameters $(a,b)$.  The P-value is the probability that $\Z_i$ is better than $\x_i$. It inherits the properties of (\ref{equ:RSS_dist}): it can  be calculated exactly without the need for simulations, data splitting or the determination of some regularization parameter, it is model  free and is valid no matter what the data.

It follows from Theorem~\ref{the:PF_PG} that $P_G(\x_i)=P_F(\x_i)$ where $P_F$ is the usual P-value based on the F-distribution. In spite of this equality the two P-values are entirely different. The randomness in the case of $P_G$ is inserted by the statistician who replaces $\x_i$ by a Gaussian covariate $\Z_i$. The randomness in the case of $P_F$ comes from the error term $\bs{\varepsilon}$ in the standard model
\begin{equation} \label{equ:mod1}
\y=\sum_{\x_j \in \mathcal{S}}^k\beta_j\x_j+\sigma\bs{\varepsilon}.
\end{equation}
with $\bs{\varepsilon}$ Gaussian noise. The Gaussian P-value $P_G(\x_i)$ is always valid, the F P-value $P_F(\x_i)$ is only valid under the model (\ref{equ:mod1}). The P-value $P_G(\x_i)$ can be calculated by simulation: simply replace $\x_i$ by $\Z_i$, run the simulations and calculate the relative frequency with which $\Z_i$ is better than $\x_i$. The P-value $P_F(\x_i)$ cannot be simulated as this would require knowledge of the true model (\ref{equ:mod1}) as well as the true values of the $\beta_j$ and $\sigma$. 

More generally given subset $\mathcal{S}$ of size $k<n$ and a covariate $\x_i \in \mathcal{S}$ its Gaussian P-value will be defined as
\begin{equation} \label{equ:P_val_2}\
P_G(\x_i)=\text{Beta}_{1,q-k+1}(\text{Beta}_{(n-k)/2,1/2}(\rss_k/\rss_{k,-i})).
\end{equation}
Given $\mathcal{S}$ of size $k<n-1$ and a covariate $\x_i \notin \mathcal{S}$ its Gaussian P-value will be defined as
\begin{equation} \label{equ:P_val_3}
 P_G(\x_i)=\text{Beta}_{1,q-k}(\text{Beta}_{(n-k-1)/2,1/2}(\rss_{k,+i}/\rss_k))
\end{equation}

Two selection procedures will be defined. The first is the all subset procedure which considers all $2^q$ subsets of the $q$ covariates and selects those subsets all of whose covariates have Gaussian P-values (\ref{equ:P_val_2}) not exceeding a specified threshold $\alpha$. The stepwise procedure is based on the Gaussian P-values  (\ref{equ:P_val_3}) where now $\mathcal{S}$ represents a the selected subset at a particular step in the procedure and a decision is to be made which if any covariates  $\x_i \notin \mathcal{S}$  are to be selected in the next step.  The Gaussian P-value (\ref{equ:P_val_3}) can be much larger than the corresponding standard F P-values. If $k=1$ and $q=176358$, an example considered below, then the Gaussian P-value 0.01 of (\ref{equ:P_val_3}) corresponds to a standard F P-value of  5.025168e-08. The Gaussian P-values derive from (\ref{equ:RSS_dist}) and inherits its properties: they are exact and valid for all data, all subset and all covariates. Gaussian white noise is the only noise for which all this holds.

An R package {\it gausscov} is available.

The remainder of this paper is organized as follows. In Section~\ref{sec:F-test.etc.2} we state Theorems~\ref{the:RSS} and \ref{the:PF_PG} from which follow (\ref{equ:RSS_dist}) and (\ref{equ:P_val_1}) and show that $P_G=P_F$ in more generality. In Section~\ref{sec:cov_selec} we define the two selection procedures, the all subsets and the step-wise procedures  and derive the selection Gaussian P-values  (\ref{equ:P_val_2})  and (\ref{equ:P_val_3}).  $\alpha$-approximation regions and intervals corresponding to $\alpha$-confidence regions and intervals are defined in Section~\ref{sec:post-select}. Section~\ref{sec:false_pos} considers the problem of false positives and false negatives. The problem of relevant groups rather than individual covariates is considered in Section~\ref{sec:rel_subset}. The dependency graphs and lagged covariates are discussed in the Sections~\ref{sec:graphs} and \ref{sec:lags}. Extensions to $M$-regression and non-linear regression and are described in Section~\ref{sec:extensions}.  Some asymptotic results on the behaviour of the step-wise procedure are given in Section~\ref{sec:boun_asymp}. Some simulation results and applications to real data sets are presented in Section~\ref{sec:examples} some of which are taken from \cite{DAV21} which gives a detailed comparison the the Gaussian covariate with 13 other covariate selection procedures.  Proofs of theoretical results and technical details are deferred to appendices.

\section{Exact probabilities for the model-free approach}
\label{sec:F-test.etc.2}
\subsection{Gaussian covariates}
 Consider a subset $\mathcal{S}$ of covariates of size $k$ and a subset $\mathcal{S}_0\subset \mathcal{S}$ of size $k_0<k$.  Regress the dependent variable $\y$ on the $\x_i\in\mathcal{S}_0$ to give sum of squared residuals of $\rss_0$. Now replace the covariates $\x_i \in \mathcal{S}\setminus \mathcal{S}_0$ by $k-k_0$ independent Gaussian covariates $\Z_i=N_n({\bf 0},{\bf I}),i=k_0+1,\ldots,k$ and regress $\y$ on the covariates $\x_i \in \mathcal{S}_0,\Z_{k_0+1},\ldots ,\Z_k$ with resulting in a sum $\RSS$ of squared  residuals. We have
 \begin{Theorem} \label{the:RSS}
\[ RSS/rss_0 \sim B((n-k)/2,(k-k_0)/2).\]
  \end{Theorem}
Theorem~\ref{the:RSS} is model free and exact whatever the data, the subsets $\mathcal{S}_0\subset \mathcal{S}$ and the covariates $\x_i\in \mathcal{S}\setminus \mathcal{S}_0$.

 The model  free approach for the combined relevance of the covariates $\x_i \in \mathcal{S}\setminus \mathcal{S}_0$ is as follows.  Regress $\y$ on all covariates $\x_i \in \mathcal{S}$ with sum of squared residuals $\rss$. The Gaussian P-value is defined by
\begin{equation} \label{equ:def_Pval}
P_G=\bs{P}(\RSS\le \rss).
\end{equation} 
It the probability that the Gaussian covariates $\Z_{k_0+1},\ldots,\Z_k$ are better than the  $\x_i \in \mathcal{S}\setminus \mathcal{S}_0$.

We have
 \begin{Theorem} \label{the:PF_PG}
The P-value (\ref{equ:def_Pval})  satisfies
 \[P_G=\ \text{Beta}_{(n - k)/2,(k- k_0)/2}(\rss/\rss_0)=1 - \text{F}_{k-k_0,n-k} \Bigl( \frac{(\rss_0 - \rss)/(k- k_0)}{\rss/(n - k)} \Bigr)=P_F.\]
 \end{Theorem}
\noindent
where $\rss_0$ denotes the sum of squared residuals for the regression based on all $\x_i \in \mathcal{S}_0$  The proof is given in the Appendix. The case $k_0=k-1$ follows from (\ref{equ:RSS_dist})  which is (\ref{equ:dist_RSS}). The general case with $k_0<k$ follows from (\ref{equ:RSS_dist_gen}).

\section{Selecting covariates}
\label{sec:cov_selec}

\subsection{All subsets}\label{sec:all_sets}

The $P_G$-value (\ref{equ:P_val_2}) is derived as follows. Given a subset $\mathcal{S}$ of size $k$ and a covariate $\x_i \in \mathcal{S}$ all the remaining covariates and $\x_i$ itself are replaced by $q-k+1$ i.i.d. Gaussian covariates $\Z_j, j=1,\ldots, q-k+1$. The sum of squared residuals based on $\mathcal{S}$ is denoted by $\rss_k$. The covariate $\x_i$ is replaced by each of the covariates $Z_j$ in turn to gives sums of squares residuals $\RSS_j, j=1,\ldots,q-k+1$. The best of the $\Z_j$ is better than $\x_i$ if $\min_{j=1,\ldots,q-k+1}\RSS_j \le \rss_k.$ Thus the Gaussian P-value of $\x_i$ is given by
\begin{equation}\label{equ:P_G_val_2}
P_G(\x_i)=\bs{P}(\min_{j=1,\ldots,q-k+1}\RSS_j \le \rss_k)=\text{Beta}_{1,q-k+1}(\text{Beta}_{(n-k)/2,1/2}(\rss_k/\rss_{k,-i}))
\end{equation}
which follows from Theorem~\ref{the:RSS}. This is (\ref{equ:P_val_2}).

The {\it gausscov} all subset function is {\it fasb}. It retains all subsets for which each covariate in the subset has a Gaussian P-value (\ref{equ:P_G_val_2}) at most $\alpha$. In a second step all subsets which are subsets of some other retained subset are discarded. The remaining subsets are maximal in the sense that it is not possible to include another covariate whilst still maintaining the upper bound $\alpha$ for all covariates in the subset. Finally the retained subsets are ordered by the sums of the squared residuals. 
 
\subsection{The Gaussian step-wise procedure} \label{sec:step_wise}

Suppose a subset $\mathcal{S}$ of $k$ covariates has already been selected with sum of squared residuals $rss_k$. There remain $q-k$ covariates. The candidate for selection is that  covariate $\x_b$ with the smallest sum of squared residuals $\rss_{k,+b}$ when $\y$ is regressed on $\mathcal{S}\cup\{\x_b\}$. Its Gaussian P-value is given by (\ref{equ:P_val_3})
\[P_G(\x_b)=\text{Beta}_{1,q-k}(\text{Beta}_{(n-k-1)/2,1/2}(\rss_{k,+b}/\rss_k))\]
by the same argument which lead to (\ref{equ:P_G_val_2}). If this is less than the cut-off value $\alpha$ then $\x_b$ is selected and the procedure continues. Otherwise the procedure terminates. The Gaussian P-values which determine whether a covariate is selected or not depend on the set of already selected covariates at this point. They can differ from the P-values calculated for the final set. If this set is not too large (specified by the user) the all subset procedure is applied to this set by default. If there is no subset all of whose Gaussian P-values are less than the cut-off value $\alpha$ the procedure terminates without specifying a subset. Otherwise that subset with the smallest  sum of squared residuals is returned. The {\it gausscov} function is {\it f1st}.

Instead of considering just one covariate for selection the first $kmn$ can be selected for a specified number  $kmn$ irrespective of the P-values.After this set has been selected the selection procedure continues until the candidate covariate has a P-value exceeding $\alpha$ when it terminates. If $kmn$ is not too large for example $kmn=20$ then all subsets of these $kmn$ covariates can be considered as in Section~\ref{sec:all_sets}. 

No step-wise procedure is guaranteed to work but Theorems \ref{thm:consistency.0}, \ref{thm:consistency.general} and \ref{thm:consistency.orthogonal} in Section~\ref{sec:boun_asymp} give sufficient condition when considering data generated under the standard linear model with a known correct set of covariates.  For large $n$ the probability of not selecting the correct subset is bounded above by $\alpha$. This supports the interpretation of $\alpha$ as an upper bound for the probability of selecting a false positive.

\subsection{Repeated Gaussian procedures} 
\label{sec:rep_step_wise}
 A selected subset $\mathcal{S}$ of covariates represents a linear approximation to the dependent variable $\y$. There will in general be more than one such approximation. Further ones may be obtained by excluding the subset selected by {\it f1st} and then applying {\it f1st} to those remaining. This is continued until no more covariates are selected by  {\it f1st}.  The {\it gausscov} function is {\it f2st}.

 A second method which is less radical than {\it f2st} is as follows. Again {\it f1st} is used to select an initial subset but now, in the second step,  instead of excluding all covariates initially selected they are excluded one at a time whilst retaining the others. {\it f1st} and then applied to those remaining. This can be iterated $m$ times where $m$ is specified by the user. The {\it gausscov} function is {\it f3st} 

\subsection{Constructing models} \label{sec:const_mod}
The Gaussian covariate selection procedures produce linear approximations and not models. However given such an approximation it is possible to construct a model making use only of the selected covariates. This is done for the riboflavin (Gaussian errors), leukemia (logit model) and sunspot (non-parametric regression with autoregressive Gaussian errors) data sets in \cite{DAV21}.

\section{$\alpha$-approximation regions and intervals} \label{sec:post-select}
The Gaussian procedure yields approximations to the dependent variable $\y$ with valid $P_G$ -values. A small modification of the $\bs{\beta}$ values will also result in an approximation although a somewhat worse one in the sense of least squares than the least squares approximation. We now consider the problem of deciding which $\bs{\beta}$ values  can be considered to gives an acceptable approximation. Given the $\bs{\beta}$ and a subset $\mathcal{S}$ of size $k$ we regress $\y-\x\bs{\beta}$ on $k$ Gaussian covariates to give a sum of squared residuals $\RSS$ and require this to be less than the least squares sum of squared residuals $\Vert y-\x\bs{\beta}\Vert^2$ . The probability that this is the case is
\begin{eqnarray}
\bs{P}(\RSS\le \Vert y-\x\bs{\beta}_{\text{ls}}\Vert^2)&=&\bs{P}(\RSS/\Vert y-\x\bs{\beta}\Vert^2\le \Vert y-\x\bs{\beta}_{\text{ls}}\Vert^2/\Vert y-\x\bs{\beta}\Vert^2)\nonumber\\
&=&\text{Beta}_{(n-k)/2,1/2}( \Vert y-\x\bs{\beta}_{\text{ls}}\Vert^2/\Vert y-\x\bs{\beta}\Vert^2)
\end{eqnarray}
from (\ref{equ:RSS_dist}) where $\bs{\beta}_{\text{ls}}$ denotes the least squares values of the $\bs{\beta}$. If we specify the probability $\alpha$ with which this is required to hold it follows after some manipulation that
\begin{equation} \label{equ:conf_reg_2}
\Vert \y-\x\bs{\beta}\Vert^2 \le \Vert \y-\x\bs{\beta}_{\text{ls}}\Vert^2/\text{Beta}^{-1}(\alpha,(n-k)/2,k/2)\}
\end{equation}
leading to the $\alpha$-approximation region
\begin{equation} \label{equ:conf_reg_2}
C(\alpha)=\{\bs{\beta}:\Vert \y-\x\bs{\beta}\Vert^2 \le \Vert \y-\x\bs{\beta}_{\text{ls}}\Vert^2/\text{Beta}^{-1}(\alpha,(n-k)/2,k/2)\}.
\end{equation}
This is the same as the standard $\alpha$-confidence regions but in contrast to the latter it is model-free and always valid.

This can be done for intervals as follows. Take the $k$th covariate $ \x_k$ with least squares coefficient $\beta_{k;\text{ls}}$. Regress $\y-(\beta_{k;\text{ls}}+\lambda)\x_k$ on the remaining $k-1$ covariates. Then the sum of the squared residuals is 
\[\Vert \y-\x\bs{\beta}_{\text{ls}}\Vert^2+\lambda^2\Vert \x_k-\text{Proj}_{k-1}(\x_k)\Vert^2=\Vert \y-\x\bs{\beta}_{\text{ls}}\Vert^2+\lambda^2\sigma_k^2\]
where  $\text{Proj}_{k-1}$ denotes the projection onto the subspace spanned by the remaining $k-1$ covariates and $\sigma^2_k=(\bs{x}^t\bs{x})^{-1}_{k,k}$. This is the increase in the sum of squared residuals using a value of $\beta_k=\beta_{k;\text{ls}}+\lambda$ which differs from the least squares value $\beta_{k;\text{ls}}$. Regress  $\y-(\beta_{k;\text{ls}}+\lambda)\x_k$ on the remaining $k-1$ covariates and a Gaussian covariate $\Z_k$ to give a sum of squared residuals $\RSS_k$. From  (\ref{equ:RSS_dist}) we have for a given $\alpha$ 
\begin{equation}    \label{equ:alpha_free_2}
\bs{P}(\RSS_k\le\text{Beta}^{-1}(\alpha,(n-k)/2,1/2)(\Vert \y-\x\bs{\beta}{_\text{ls}}\Vert^2+\lambda^2\sigma_k^2))=\alpha
\end{equation}
so that $\bs{P}(\RSS_k\le\Vert \y-\x\bs{\beta}_{\text{ls}}\Vert^2)\ge \alpha$ if
\[\lambda^2\le\frac{\Vert\y-\x\bs{\beta}_{\text{ls}}\Vert^2}{\sigma_k^2}\left(\frac{1}{\text{Beta}^{-1}(\alpha,(n-k)/2,1/2)}-1\right)\]
which corresponds to the standard $1-\alpha$ confidence interval based on the t-distribution.

\section{False positives and false negatives} \label{sec:false_pos}
False positives and false negatives are usually defined in terms of hypotheses about parameter values in a linear regression. A false positive is the rejection of the hypothesis $H_j:\beta_j=0$ although it is true, a false negative is the acceptance of $H_j$ although it is false. In simulations these definitions can be used and can be of interest. For real data matters are more complicated and the decisions can only be made on knowledge of the data. 

A false positive would be a covariate which is included in the selection but has no relevance for the dependent variable $\y$, a case of a spurious correlations, for example where $\y$ and the covariate $\x$ increase over time. If care has been taken with the data so that no clearly irrelevant covariates have been included then the Gaussian covariate procedure will avoid false positive. Any selected covariate has a Gaussian P-value of less than $\alpha$ which means that it is significantly better than i.i.d. Gaussian covariates .This means that the Gaussian covariate procedure does  not overfit, a property confirmed in practice (see the simulations and examples in \cite{DAV21})).

False negatives are more difficult. A false negative is a relevant covariate which is relevant but is not selected.  This can happen in multiple ways, a non-linearity in the relationship between the dependent variable and one or more covariates, the step-wise procedure failing because the first Gaussian P-value exceeds the cut-off value, an inability to consider all subsets when $q$ is large. This latter problem can be mitigated to some extent as described in Sections~\ref{sec:step_wise} and \ref{sec:rep_step_wise}. A claim about false negatives is more difficult to make than one about false positives as it involves a statement about a relevant covariate existing although its existence cannot be established.

\section{Relevant groups}\label{sec:rel_subset}
It can happen that a group of covariates is relevant although the the effect of the individual covariates is not sufficiently strong for this to be detected. The group lasso was proposed in \cite{YUANLIN06} to try and identify such groups ( see also Section~4 of \cite{DEBUEMEME15}).  We consider here the case that the Gaussian P-values exceed the cut-off value $\alpha$ but the P-value of the $R^2$ statistic is small in a sense to be made clear indicating that the covariates taken as a whole do have a  relevant effect. So far we have only come across this problem in the simulations  in Sections~\ref{sec:sim1} and \ref{sec:rangrph}. The reason seems to be that in these simulations all the covariates are Gaussian and all the $\beta_i$ are the same. 

As an example we take the simulations discussed in Section~\ref{sec:sim1}. The parameters are $(n,q)=(1000,1000)$ and 60 of the covariates have a non-zero coefficient value, namely $\beta=4.5/\sqrt{1000}$. We use the step-wise Gaussian method to choose 60 covariates. In one such simulation default version of the Gaussian method 54 of these had non-zero coefficients but the P-values of only nine covariates were below the cut-off values of which eight had a non-zero coefficient.   The sum of the squared residuals was 888.65 based on all 60. We now regress the dependent variable $\Y_{1000}$ on 1000 covariates generated generated in the same manner but with all coefficients zero. Of these the first 60 were chosen using using the Gaussian step-wise procedure with $kmn=60$ as in Section~\ref{sec:step_wise} and the dependent variable regressed on these 60. Over 500 such simulations the smallest sum of squared residuals was  1099 giving a P-value so to speak of 0.  Repeating this with $\beta=1/\sqrt{1000}$  gave a P-value of 0.2 indicating that this value of $\beta$ is  about the limit of detectability.

We propose the following. The default step-wise method compares the best of the remaining covariates with the best of the same number of i.i.d. $N_n(\boldsymbol{0},\boldsymbol{I})$ which is the first order statistic. We weaken this by comparing the best of the remaining covariates with the $\nu$th best of the random Gaussian covariates. If a subset of size $k$ has already been selected the Gaussian P-value of the next best covariate $\x_b$ is defined as 
\begin{equation} \label{equ:step_P_nu}
P_G(\x_b)=B_{\nu,q-k+2-\nu}(B_{(n-k-1)/2,1/2}(\rss_{k,+b}/\rss_k))
\end{equation}
where we use the same notation as for (\ref{equ:P_val_2}). Again, this probability is exact. One could instead just specify another cut-off probability instead of the default value $\alpha=0.01$ but is not easily interpretable which is why we prefer specifying $\nu$.

The larger $\nu$ the more likely it is that false positives will be selected. To estimate the number of false positives we regress $\y\in\R^n$, any $\y\ne 0$ as it is model-free, on $q$ i.i.d. $N_n(\boldsymbol{0},\boldsymbol{I})$ Gaussian covariates for a given $\nu$. Any selected covariate is a false positive. The package {\it gausscov} contains a function {\it fnfp} which gives values for $50\le n\le 5000,\,25\le q\le 50000,\,1\le nu\le 10, \,\alpha\in \{0.01,0.05\}$ by interpolating the results obtained from previous simulations. Other values can be simulated. As an example we put $(n,q,\alpha,\nu)=(1000,1000,0.01,\text{c}(5,10))$ which is used in Section~\ref{sec:sim1}. The means obtained from interpolating previous are 1.345 and 4.615. If simulations are used the means and standard deviations and a histogram are returned. The results of of 5000 simulations using  {\it fnfp} are given in Table~\ref{tab1}. The means for $\nu=5$ and $\nu=10$ are  1.295   and 4.571 and the standard deviations 1.19 and 2.41  respectively.
{\footnotesize
\begin{table}[h] 
\begin{center}
\begin{tabular}{ccccccccccccccc}
$\nu$&0&1&2&3&4&5&6&7&8&9&10&$\ge$ 11\\
1&0.99& 0.01&0.00&0.00&0.00&0.00&0.00&0.00&0.00&0.00&0.00&0.00\\
5&0.29& 0.34& 0.22&0.15&0.10&0.03&0.01&0.01&0.00&0.00&0.00&0.00\\
10&0.02&0.06&0.12&0.16&0.17&0.15&0.11&0.09&0.05&0.03&0.02&0.02\\
\end{tabular}
\caption{Histogram of false positives
  $(n,q,\alpha)=(1000,1000,0.01)$  with $\nu=1,\,5$ and $10$ based on 5000 simulations
  using {\it fnfp}. \label{tab1}}
\end{center}
\end{table}
}

Thus increasing $\nu$ from 1 to 10 will on average lead to about  about 4.6 false positives. If the increase in the number of covariates selected is much greater than this it may be deemed reasonable to use $\nu=10$. Examples of this are given in the simulations in Sections~\ref{sec:sim1} and \ref{sec:rangrph} .

\section{Graphs and lagged covariates}
\label{sec:graphs_lag}
One major advantage of the model-free nature of the covariate selection procedures is that they can be applied without change to situations which are modelled  in very different ways. We give two examples, the construction of graphs and the use of lagged covariates.

\subsection{Graphs}
\label{sec:graphs}
Given the model
\begin{equation} \label{equ:graph_mod}
\X=(\X_1,\ldots,\X_k)\sim \mathcal{N}[\bs{\mu},\bs{\Sigma})
\end{equation}
with $k<n$ the graphical independence structure of the distribution can be obtained from the location of zeros in the inverse matrix $\bs{\Sigma}^{-1}$ (\cite{WHITT90}). The structure can also be obtained by regressing each $\X_i$ on the remaining $\X_j$. This approach can be extended to the case $k>n$ using covariate selection methods as is shown in \cite{MEIBUE06}.

A graph can be constructed as follows. Each covariate $\x_i$ is regressed on the remaining covariates using the step-wise Gaussian covariate method. The covariate $\x_i$ is then joined to the selected covariates $\x_{\ell},\ell \in S_i$ to by edges with arrows pointing from the $\x_j$ to $\x_i$ to denote that the $\x_i$  depends on the $\x_j$. This gives a directed graph. An undirected graph is composed of edges without arrows to denote that the covariates are related.  In the default version the cut-off values $\alpha$ is set to $\alpha/q$ where $q$ is the number of covariates. The {\it f2st} or {\it f3st} of Section~\ref{sec:rep_step_wise} can also be used and typically give much larger graphs. 

\subsection{Lagged covariates}
\label{sec:lags}
 Modelling and analysing a data set using models based on lagged data is not simple involving as it does the determination of the coefficients and the order of the lags involved. Furthermore it seems only to be possible to do such an analysis if the order is small. The Gaussian step-wise selection procedure avoids these problems and can used to analyse vector lagged covariates. This is done for some American Business Cycle data in Section~\ref{sec:lag_cov}.

\section{Beyond least squares}
\label{sec:extensions}

We briefly consider extension to robust ($M$-)regression and non-linear regression. 

\subsection{$M$-regression} \label{sec:rob}
Let $\rho$ by a symmetric, positive and twice differentiable convex function with $\rho(0)=0$.  The default function will be the Huber's $\rho$-function with a tuning constant $c$ (\cite{HUBRON09}, page 69) defined by
\begin{equation}
\rho_{c}(u)=\left\{\begin{array}{ll}
\frac{u^2}{2}, &\vert u\vert \le c,\\
c\vert u\vert -\frac{c^2}{2},&\vert u \vert > c.\\
\end{array}
\right.
\end{equation}
The default value of $c$ will be $c=1$.

For a given subset ${\mathcal S}$ of covariates of size $k$ the sum of squared residuals is replaced by
\begin{equation}  \label{equ:min_rho}
s_0(\rho,\sigma)=\min_{\bs{\beta}({\mathcal S})}\,\frac{1}{n}\sum_{i=1}^n 
\rho\left(\frac{y_i-\sum_{j \in{\mathcal S}}x_{ij}\beta_j({\mathcal
      S})}{\sigma}\right).
\end{equation}
which can be calculated using the algorithm described in {\bf 7.8.2} of \cite{HUBRON09}.  The minimizing $\beta_j({\mathcal S})$ will be denoted by $\beta_j({\mathcal S},\text{lr})$.

For some $\x_\nu\notin {\mathcal S}$  put  
\begin{equation}  \label{equ:min_rho_j}
s_{\nu}(\rho,\sigma)=\min_{\bs{\beta}({\mathcal S}\cup \{\x_\nu\})}\,\frac{1}{n}\sum_{j=1}^n 
\rho\left(\frac{y_j-\sum_{j \in{\mathcal S}\cup\{\x_\nu\}}x_{ij}\beta_j({\mathcal
      S}\cup\{\x_\nu\})}{\sigma}\right).
\end{equation}

 Replace all the covariates not in ${\mathcal S}$ by standard  Gaussian white noise, include the $\ell$th such random covariate  denoted by $Z_{\ell}$ and put
\begin{equation} 
S_{\ell}(\rho,\sigma)=\min_{\bs{\beta}({\mathcal
    S}),b}\,\frac{1}{n}\sum_{i=1}^n   
\rho\left(\frac{y_i-\sum_{j \in{\mathcal S}}x_{ij}\beta_j({\mathcal 
      S})-bZ_{\ell}}{\sigma}\right).
\end{equation}
A Taylor expansion gives
\begin{eqnarray}
S_{\ell}(\rho,\sigma)&\approx&\frac{1}{2}\frac{\left(\sum_{i=1}^n\rho^{(1)}\left(\frac{r_i}{\sigma}\right)Z_i\right)^2}{\sum_{i=1}^n\rho^{(2)}\left(\frac{r_i}{\sigma}\right)Z_i^2}\nonumber\\
&\approx&s_0(\rho,\sigma)-\frac{1}{2}\frac{\left(\sum_{i=1}^n\rho^{(1)}\left(\frac{r_i}{\sigma}\right)\right)^2}{\sum_{i=1}^n\rho^{(2)}\left(\frac{r_i}{\sigma}\right)}\chi^2_1
\end{eqnarray}
with $r_i=y_i-\sum_{j \in{\mathcal S}}x_{ij}\beta_j({\mathcal S},\text{lr})$.
This leads to the asymptotic $P$-value  for ${\bf x}_{\nu}$
\begin{equation} \label{equ:pval_m}
1-\text{Chisq}\left(
  \frac{2s_0(\rho^{(2)},\sigma)}{s_0(\rho^{(1)},\sigma)}
  (s_0(\rho,\sigma)-s_{\nu}(\rho,\sigma))\right)^{q-k}.   
\end{equation}
corresponding to the exact Gaussian $P$-value (\ref{equ:P_val_3}) for the step-wise procedure. The P-value corresponding to the exact Gaussian P-value  (\ref{equ:P_val_2}) for the all subsets procedure is obtained by replacing $k$ by $k-1$.  Here 
\[ s_0(\rho^{(1)},\sigma)= \frac{1}{n}\sum_{i=1}^n 
\rho^{(1)}\left(\frac{r_i}{\sigma}\right)^2 , \quad
s_0(\rho^{(2)},\sigma)= \sum_{i=1}^n  
\rho^{(2)}\left(\frac{r_i}{\sigma}\right).\]

It remains to specify the choice of scale $\sigma$. The initial value of $\sigma$  is the median absolute deviation of $\bs{y}$ multiplied by the Fisher consistency factor 1.4826.  After the next covariate has been  included the new scale $\sigma_1$ is taken to be
\begin{equation} \label{equ:sig_rob_reg}
\sigma_1^2=\frac{1}{(n-k-1)c_f}\sum_{i=1}^n \rho^{(1)}(r_1(i)/\sigma_0)^2
\end{equation}
where the $r_1(i)$ are the residuals based on the $k+1$ covariates and $c_f$ is the Fisher consistency factor given by
\[c_f=\bs{E}(\rho^{(1)}(Z)^2) \]
where $Z$ is ${\mathcal N}(0,1)$ (see \cite{HUBRON09}). Other choices are possible.

\subsection{Non-linear approximation} \label{sec:non_lin}
For a given subset ${\mathcal S}$ of covariates of size $k$ the dependent variable $\bs{y}$ is now approximated by $g(\x({\mathcal S})\bs{\beta}({\mathcal S}))$ where $g$ is a smooth function. Write 
\begin{equation}
ss_0= \min_{\bs{\beta}({\mathcal S})}\,\frac{1}{n}\sum_{i=1}^n (y_i-g(\x_i({\mathcal S})^{\top}\bs{\beta}({\mathcal S})))^2.
\end{equation}
and denote the minimizing $\bs{\beta}({\mathcal S})$ by $\bs{\beta}({\mathcal S},\text{ls})$. Now include one additional covariate $\x_{\nu}$ with $\bs{x}_{\nu} \notin {\mathcal  S}$ and denote the mean sum of squared residuals by $ss_{\nu}$. As before all covariates not in ${\mathcal S}$ are replaced by standard Gaussian white noise. Include the $\ell$th random covariate denoted by $Z_{\ell}$ and put
\[SS_{\ell}=\min_{\beta({\mathcal S}),b}\frac{1}{n}\sum_{i=1}^n (y_i-g(\x_i({\mathcal S})^{\top}\bs{\beta}({\mathcal S})+bZ_{\ell}))^2.
\]

Arguing as above for robust regression results in 
\begin{equation} \label{equ:lsq_non_lin_1}
SS_1\approx ss_0- \frac{\sum_{i=1}^n r_i({\mathcal
     S})^2g^{(1)}(\x_i({\mathcal S})^{\top}\bs{\beta}({\mathcal S},\text{ls}))^2}{\sum_{i=1}^ng^{(1)}(\x_i({\mathcal S})^{\top}\bs{\beta}({\mathcal S},\text{ls}))^2}\chi^2_1
\end{equation}
where 
\begin{equation}\label{equ:lsq_non_lin_2}
r_i({\mathcal S})=y_i-g(\x_i({\mathcal S})^{\top}\tilde{\bs{\beta}}({\mathcal S},\text{ls})).
\end{equation}
The asymptotic $P$-value for the covariate ${\bf x}_{\nu}$ corresponding to the asymptotic $P$-value (\ref{equ:pval_m}) for $M$-regression is  
\begin{equation} \label{equ:pval_non_lin}
1-\text{Chisq}\left(\frac{(ss_0-ss_{\nu})\sum_{i=1}^ng^{(1)} (\x_i({\mathcal S})^{\top}\bs{\beta}({\mathcal
      S},\text{ls}))^2} {\sum_{i=1}^n r_i({\mathcal
     S})^2g^{(1)}(\x_i({\mathcal S})^{\top}\bs{\beta}({\mathcal S},\text{ls}))^2},1\right)^{q-k}.
\end{equation}

In the case of logistic regression with $g(u)=\exp(u)/(1+\exp(u))$ we have
\begin{equation}\label{equ:lsq_non_lin_2_logistic}
\frac{\sum_{i=1}^n r_i({\mathcal
     S})^2g^{(1)}(\x_i({\mathcal S})^{\top}\bs{\beta}({\mathcal
     S},\text{ls}))^2}{\sum_{i=1}^ng^{(1)}(\x_i({\mathcal
     S})^{\top}\bs{\beta}({\mathcal S},\text{ls}))^2}
 = \frac{\sum_{i=1}^n(y_i-p_i(0))^2p_i(0)^2(1-p_i(0))^2}
 {\sum_{i=1}^np_i(0)^2(1-p_i(0))^2}  
\end{equation}
where 
\[p_i(0)=\frac{\exp(\x_i({\mathcal S})^{\top}\bs{\beta}({\mathcal
    S},\text{ls}))}{1+\exp(\x_i({\mathcal S})^{\top}\bs{\beta}({\mathcal
    S},\text{ls}))}.\]
This corrects a mistake in Chapter 11.6.1.2 of \cite{DAV14} where 
\[\frac{\sum_{i=1}^np_i^3(1-p_i)^3} {\sum_{i=1}^np_i^2(1-p_i)^2}\]
occurs repeatedly instead of
\[\frac{\sum_{i=1}^n(y_i-p_i)^2p_i^2(1-p_i)^2}
 {\sum_{i=1}^np_i^2(1-p_i)^2}.\]

\section{Bounds and asymptotics}
\label{sec:boun_asymp}
 We provide some theoretical results about the step-wise choice of covariates in the model-based framework, in Tukey's sense a `challenge'. Throughout this section we assume that
 \[
 	\y \ = \ \bmu + \sigma \Z
 \]
 with unknown parameters $\bmu \in \R^n$, $\sigma > 0$ and random noise $\Z \sim N_n(\bs{0},\bs{I})$. Moreover, we assume without loss of generality that $\|\x_i\|=1, i=1,\ldots,q$. The set of chosen covariates is denoted by $\hat{{\mathcal S}}$ .

We consider firstly the case of no signal, $\bmu = \bs{0}$. In this situation the correct decision is $\hat{{\mathcal S}}=\emptyset$. 
\begin{Theorem}
\label{thm:consistency.0}
If $\bmu = \bs{0}$ then 
\[
	\bs{P}(\hat{{\mathcal S}} \ne \emptyset) \ \le \ - \log(1 - \alpha) .
\]
Furthermore if $q \to \infty$ and $n/\log(q)^2 \to \infty$ then for fixed $\alpha \in (0,1)$,
\[
	\bs{P}(\hat{{\mathcal S}} \ne \emptyset)
	\ \le \ \alpha + o(1)
\]
as  uniformly in $(\x_i), i=1,\ldots,q.$. In the special case of orthonormal regressors $\x_i$,
\[
	\bs{P}(\hat{{\mathcal S}} \ne \emptyset)
	\ \to \ \alpha
\]
$q \to \infty$.
\end{Theorem}

If $\bmu \ne \bs{0}$ we suppose that $\bmu=\sum_{\x_i \in {\mathcal S}_{*}}\beta_i\x_i$ where $ {\mathcal S}_{*}$ is a subset of size $k_{*}<n$ and the $\x_i\in  {\mathcal S}_{*}$ are linearly independent. For any subset ${\mathcal S}$ we denote the linear subspace of $\R^n$ spanned by the $\x_i\in{\mathcal S}$ by $\V_{{\mathcal S}}$ and the orthogonal complement of this subspace by $\V_{{\mathcal S}}^{\perp}$. The orthogonal projection onto $\V_{{\mathcal S}}^\perp$ is denoted by $Q_{\mathcal S}$  and for any $\x_i \notin {\mathcal S}$ we write
\[
	\x_{{\mathcal S},i}^{} \ := \ \|Q_{\mathcal S}\x_i\|^{-1} Q_{\mathcal S} \x_i
\]
(with $0^{-1} \bs{0} := \bs{0}$).

With the above notation we have
\begin{Theorem}[Consistency of step-wise choice, general design]
\label{thm:consistency.general}
Suppose that
\[
	\bmu \ \in \ \V_{{\mathcal S}_*}
\]
and that the two following assumptions hold:\\
(A.1) \ $\min(n,q)/k_* \to \infty$ and $\log(q)^2/n \to 0$, and

\noindent
(A.2) \ for some fixed $\tau > 2$,
\[
	\min_{\x_j \in {\mathcal S}_*, {\mathcal S} \subset {\mathcal S}_* \setminus \{\x_j\}, \x_i \notin {\mathcal S}_*}
		\, \frac{|\x_{{\mathcal S},j}^\top\bmu| - |\x_{{\mathcal S},i}^\top\bmu|}
			{\sqrt{n \sigma^2 + \|\bmu\|^2}}
	\ \ge \ \frac{\sqrt{\tau \log q} + 2\sqrt{k_*}}{\sqrt{n}} .
\]
Then the step-wise procedure yields a random set $\hat{{\mathcal S}}$ such that
\[
	\bs{P}({\mathcal S}_* \subset \hat{{\mathcal S}}) \ \to \ 1
	\quad\text{and}\quad
	\bs{P}({\mathcal S}_* \subsetneq \hat{{\mathcal S}}) \ \le \ \alpha + o(1) ,
\]
\end{Theorem}

If the $\x_i\in {\mathcal S}_*$ are orthonormal the result can be simplified.

\begin{Theorem}[Consistency of step-wise choice, orthogonal design]
\label{thm:consistency.orthogonal}
Suppose 
\[
	\bmu \ = \ \sum_{i\in{\mathcal S}_*} \beta_i \x_i
\]
where the $\x_i$ are orthonormal and that the two following conditions hold\\
\noindent
(A.1') \ $q/k_* \to \infty$, and

\noindent
(A.2') \ for some fixed $\tau > 2$,
\[
	\min_{i\in {\mathcal S}_*}
		\, \frac{|\beta_i|}{\sqrt{n\sigma^2 + \sum_{\x_i \in{\mathcal S}_*}\beta_i^2}}
	\ \ge \ \frac{\sqrt{\tau \log q} + \sqrt{2 \log k_*}}{\sqrt{n}} .
\]
Then step-wise procedure yields a random set $\hat{{\mathcal S}}$ such that
\[
	\bs{P}({\mathcal S}_* \subset \hat{{\mathcal S}}) \ \to \ 1
	\quad\text{and}\quad
	\bs{P}({\mathcal S}_* \subsetneq \hat{{\mathcal S}}) \ \le \ \alpha + o(1) .
\]
\end{Theorem}

It  is of interest to compare Theorem~\ref{thm:consistency.orthogonal} with Theorem 1 of \cite{LOKTAYTIB214} for lasso regression. There they prove (in our notation) that the first $m_*$ covariates entering the lasso path are, with probability tending to 1, those in ${\mathcal S}_*$. Our condition (A.2')  is replaced by the weaker
 \[\min_{\x_i\in {\mathcal S}_*}\, \vert\beta_i\vert-\sigma \sqrt{2\log(q)} \rightarrow \infty.\]
 However their result is restricted to $q<n$, they use the given $\sigma$, not an estimate, and there is no termination rule. See their Remark 1 on page 420 and their Section 6.

\section{Simulations and real data}
\label{sec:examples}
A detailed comparison of {\it gausscov} with the following 13 selection procedures is given in \cite{DAV21}: lasso (\cite{TIB96}), knockoff (\cite{CAFAJALV2018}), scaled sparse linear regression (\cite{SUNZHA12}), SIS (Sure Independence Screening) (\cite{FANLV08}), desparsified lasso (\cite{ZHAZHA14}), stability selection (\cite{MEIBUE10}), ridge regression (\cite{BUEH13}), multiple splitting (\cite{WASSROED09}), EMVS (Expectation-Maximization Approach to Bayesian Variable Selection) (\cite{ROCGEO14}) and   Spike and Slab Regression (\cite{SCO21}), Threshold Adaptive Validation (\cite{LFL21}), graphical lasso (\cite{FHT08,FRHATI19})  and huge (High-Dimensional Undirected Graph Estimation) (\cite{JFLRLWLZ21}.  

The comparison is based on two simulations and the following seven data sets: riboflavin \cite{BUEKALMEI14}, leukemia \cite{GOLETAL99}, lymphoma \cite{ALI00} and \cite{DETBUH02}, osteoarthritis \cite{COXBATT17}, the Boston Housing data set \cite{HARRUB78} , sunspot data \cite{SIDB}  and the American Business Cycle data \cite{ABC86}.  All the comparisons were done using R version 4.1.2 (2021-11-01) and the package {\it gausscov} with the default values  for $\alpha=0.01$ and $kmn=10$.

\subsection{Simulations} \label{sec:sims}
 \subsubsection{Tutorial 1} \label{sec:sim1}
The knockoff procedure is explained in \cite{CAFAJALV2018}. The tutorial in question is Tutorial 1 of
{\footnotesize
\begin{verbatim}
https://web.stanford.edu/group/candes/knockoffs/software/knockoff/  
\end{verbatim}
}
\noindent
which gives a simulation using knockoff. The dimensions are $(n,q)=(1000,1000)$. The 1000 covariates are Gaussian and dependent with a Toeplitz covariance matrix $\Sigma$ given by $\Sigma_{i,j}=\rho^{\vert i-j\vert}$ with $\rho=0.25$. Of the covariates $p=60$ are chosen at random and denoted by $\X_i,i=1,\ldots,60$. The dependent variable $\Y$ is given by  
\[\Y=\sum_{i=1}^{60}\beta_iX_i+N_{1000}(\boldsymbol{0},\boldsymbol{I})\]
with all the $\beta_i=amplitude/\sqrt{n}$ with $amplitude=4.5$. These are the particular values chosen for the first simulation discussed below. There is a second tutorial with a binary dependent variable. The results are similar and not given here but are available in \cite{DAV18} with however $\alpha=0.05.$ 
\begin{table}[h]
\begin{center}
{\footnotesize
\begin{tabular}{ccccc}
&\multicolumn{2}{c}{Tutorial 1}\\
method&fp&fn&time\\
\hline
lasso&68.7&1.5&12.6\\
knockoff&6.8&10.4&74.1\\
$\nu=1$&0.0&53.1&0.05\\
$\nu=5$&2.5&14.5&0.19\\
$\nu=10$&5.6&7.5&0.23\\
\hline
\quad
\end{tabular}
}
\caption{Comparison of lasso,  knockoff and Gaussian
    covariates based on 10 simulations with
    $(n,q,p,amplitiude,\rho)=(1000,1000,60,4.5,0.25)$. \label{tab2}}
\end{center}
\end{table}

The number of false positives is denoted by `fp' and false negatives by `fn'. The total number of covariates selected is given by 60-fn+fp. The time for each simulation is given in seconds.  The first line for lasso shows that on average it selects about 130 covariates each selection requiring about 12 seconds. Almost all the relevant covariates are chosen but also on average about 70 false ones. Knockoff selects on average about 60 covariates of which about 7 are false positives. It requires about 74 seconds for each selection. The Gaussian covariate method with default value $\alpha=0.01$ selects on average just 7 covariates. None of  these are false positives. Putting $\nu=5$ results in $60-14.5+2.5\approx 48$ covariates being selected. To judge how many of these are false positives we use {\it fnfp} as described in Section~\ref{sec:rel_subset}. As $fnpf(1000,1000,0.01,c(5,10), nufp)=c(1.345,4.615)$ we expect about 1.5 false positives if $\nu=5$ and about $4.6$ if $\nu=10$.  These numbers agree with the Table~\ref{tab1} derived from simulations and also with the values in Table~\ref{tab2}. Thus in terms of minimizing the number of false decisions $\nu=10$ would seem to be the best choice.  We emphasize here that the choice $\nu=10$ results from using {\it fnfp} and not by choosing the best value on running Tutorial 1.

\subsubsection{Random graphs} \label{sec:rangrph}
This is based on \cite{MEIBUE06} but with $(n,q)=(1000,600)$.  On the last line of page~13 of  \cite{MEIBUE06} the expression $\varphi(d/\sqrt{p})$ with $\varphi$ the density of the standard normal distribution and $d$ the Euclidean distance is clearly false. It has been replaced by $\varphi(23.5d)$ which gives about 1800 nodes compared with the 1747 of \cite{MEIBUE06}. The Meinshausen-B\"uhlmann method with $\alpha=0.05$ and non-directed edges resulted in  1109 edges of which two were false positives giving 640 false negatives.  

 One simulation of the modified (as described above) Meinshausen-B\"uhlmann random graph method produced 1823 edges. The Gaussian method described in Section~\ref{sec:graphs} yielded 1590 edges of which two were false positive and  235 were false negatives. The time required was about 9 seconds.  

Putting $\nu=2$ resulted 1821 edges, that is 231 more than with $\nu=1$. Using {\it fnfp} with $p=0.01$, $\nu=2$, $gr=T$ and $nsim=10^5$ the mean number of false positives per covariate was $0.00915$ suggesting a Poisson distribution with mean  5.5 for the number of false positives.  Thus of the 231 additional edges one can expect that between one and 12 are false positives. The actual number was nine with 11 false negatives.

In \cite{DAV21} the Gaussian covariate procedure is compared with the following three procedure for constructing dependency graphs: Threshold Adaptive Validation (\cite{LFL21}),  huge (High-Dimensional Undirected Graph Estimation) (\cite{JFLRLWLZ21} and graphical lasso (\cite{FHT08,FRHATI19}). The graph was constructed as above but with $(n,q)=1000$. Table~\ref{tab:rnd_graph} is Table 11 of  \cite{DAV21} with time measured in seconds. 
\begin{table}[ht]
\begin{center}
\begin{tabular}{rcccc}
\multicolumn{5}{c}{Random graph (1000,1000)}\\
\hline\\
method&no. edges&$fp$&$fn$&time\\
\hline\\
{\it fgr1st} &1820&1&3&27.2\\
{\it thav.glasso} &1776&218&265&90\\
{\it huge}&1839&30&14&25.5\\
{\it glasso}&1840&293&276&14.1\\
\end{tabular}
\caption{The results for one simulation of the random graph. \label{tab:rnd_graph}}
\end{center}
\end{table}

\subsubsection{Riboflavin simulations}
The following is taken from \cite{DAV21}. The riboflavin covariates are standardized to have mean zero and variance one.  Four covariates $\{\x_{i_1},\x_{i_2},\x_{i_3},\x_{i_4}\}$ are chosen at random and the dependent variable $\Y$ generated as 
\[\Y=20\sum_{j=1}^4 \x_{i_j} +\bs{\varepsilon}\]
where $\bs{\varepsilon}$ is standard Gaussian noise. Table~\ref{tab:ribo_sim} gives the results of 100 simulations.  
\begin{table}[h]
\begin{center}
\begin{tabular}{rcccc}
\multicolumn{5}{c}{Riboflavin: 100 simulations; (*) 70, (**) 72, (***) 18 simulations }\\
\hline\\
method&$fp$&$fn$&\% correct&time\\
{\it f1st}&0.77&0.72&75&1 (0.026)\\
{\it f3st,m=1}& 0.18& 0.17&93& 5\\
{\it f3st,m=2}& 0.07&0.05&98&24\\
{\it lasso}&25.0& 0.07& 0&19\\
{\it scalreg}& 16.3&1.08& 0&85\\
{\it SIS} & 13.5& 2.45&3&150\\
{\it stability}& 0.24& 2.16&8&96\\
{\it multi-split(*)}&0.23&1.59&36&1570\\
{\it BoomSpikeSlab(**)}& 0.33& 0.42&87&1540\\
{\it EMVS} & 0.00& 4.00& 0&27\\
{\it knockoff}& ?& ?&?&$>$150000\\
{\it desparse.lasso}& ?& ?&?&$>$150000\\
{\it ridge(***)}& 0.00&4.00&0.00&10000\\
\end{tabular}
\caption{Columns 2-4 give the average number of false positives, the average number of false negatives and the \% of correct selections. Column 5 gives the time compared with Gaussian covariates which required on average 0.026 seconds per simulation. \label{tab:ribo_sim}}
\end{center}
\end{table}

\subsection{Real data}
\label{sec:real_data}
 
 \subsubsection{Riboflavin data} \label{sec:ribp}
Table~\ref{tab:riboflavin} is taken from \cite{DAV21} and gives the results of applying the ten model based procedures to the riboflavin data. This particular data set has proved difficult for model based procedures, see  \cite{DEBUEMEME15} and \cite{LOCK17}. Table~\ref{tab:riboflavin} gives the results of applying the ten model based procedures to the riboflavin data. The columns are the procedures, the number of selected covariates and false positives $(k,fp)$, whether P-values are given the sum of squared residuals $\ss$ and the time as compared with {\it f1st} which took 0.024 seconds. A false positive is defined as a covariate with a Gaussian P-value exceeding 0.99.

Table~\ref{tab:f3st} gives the first five approximations of the 129 yielded by {\it f3st} with $kmn=15$ and $m=5$. The first line of Table~\ref{tab:riboflavin} was number 37 on the list.
\begin{table} [h]
\begin{center}
\begin{tabular}{rcccc}
\multicolumn{5}{c}{riboflavin (71,4088)}\\
\hline\\
method& $k,fp$&P-values&$\bs{ss}$&time\\
\hline\\
{\it f1st}&4,0&yes& 8.45&1 (0.024)\\
{\it f3st,m=1}&6,0&yes&6.21&4\\
{\it lasso}&32,30&no&2.05&25.7\\
{\it knockoff}&0,0&no&*&$>$7e+05 (killed)\\
{\it scalreg}&9,6&no&10.62&28.7\\
{\it SIS}&4,0&no&11.49&89\\
{\it desparsified lasso}&0,0&yes&*&130012\\
{\it stability}&0,0&no&*&103\\
{\it ridge.proj}&0,0&yes&*&12248\\
{\it multi-split}&4,2&yes&17.45&1421\\
{\it EMVS}&0,0&no&*&22\\
{\it BoomSpikeSlab}&(5,2)&no&16.92&2290\\
\end{tabular}
\caption{The results for the riboflavin data. \label{tab:riboflavin}}
\end{center}
\end{table}

\begin{table}[hb]
\begin{center}
\begin{tabular}{cccccccccc}
$\bs{ss}$&\multicolumn{9}{c}{Riboflavin: Included covariates}\\
\hline
 3.72& 4004& 2564&   73&  315& 2936 & 997&  991& 1661&  3255\\
4.23& 4004& 2564 &  73 & 315& 2936&  997& 1661& 2048&*\\
  4.87& 4004& 2564 &144& 1131& 3138& 2186& 1771&*&*\\
  5.43& 1279& 4004& 2564&   73& 1131& 2140&    *&*&*\\
  5.47& 4003& 2564 &  69&1425&  413& 2484& 1194&*&*\\
\end{tabular}
\caption{The first five of the 129 approximations given by {\it f3st} with $kmn=15$ and $m=5$ in order of the sum of squared residuals $\bs{ss}$.  \label{tab:f3st}}
\end{center}
\end{table}

\subsubsection{Lagged covariates} \label{sec:lag_cov}

The American Business Cycle data we considered are the USA quarterly data 1919-1941,1947-1983 available from
\begin{verbatim}
http://data.nber.org/data/abc/
\end{verbatim}
We merged the two time intervals and used the values given in
1972\$. The dependent variable was taken to be the Gross national
Product (GNP72).   The following further indices (see the above data source for
an explanation) were included each with lags of 1:16 giving 352 covariates in all:\\
\noindent
CPRATE, CORPYIELD, M1, M2, BASE, CSTOCK, WRICE67, PRODUR72,
NONRES72, IRES72, DBUSI72, CDUR72, CNDUR72, XPT72, MPT72, GOVPUR72,
NCSPDE72, NCSBS72, NCSCON72,CCSPDE72,CCSBS72\\
\noindent
We are not economists so whether this makes sense or not we leave to
the reader. The Gaussian step-wise procedure in Table~\ref{tab:abcq_lag} selected the covariates 1,18,180 which are lag 1 of GNP72, lag 2 of CPRATE and lag 4  of IRES72.
\begin{table} [ht]
\begin{center}
\begin{tabular}{rcccc}
\multicolumn{5}{c}{American Business Cycle (224,352)}\\
\hline\\
method& $k,fp$&P-values&$\bs{ss}$&time\\
\hline\\                     
{\it f1st}&3,0&yes& 18765&1 (0.039)\\
{\it f3st,m=1}&6,0&yes&18405&3\\
{\it lasso}&4,2&no&24980&3\\
{\it scalreg}&83,69&no&4960&19\\
{\it SIS}&5,0&no&17854&16\\
{\it desparsified lasso}&190,189&yes&40&1000\\
{\it stability}&2,0&no&25460&12\\
{\it ridge.proj}&103,97&yes&8130&30\\
{\it multi.split}&2,0&yes&25460&200\\
{\it EMVS}&223, NaN&no&0&2.3\\
{\it BoomSpikeSlab}&(4,0)&no&48750&65\\
\end{tabular}
\caption{The results for the American Business Cycle data with lags 1:16. \label{tab:abcq_lag}}
\end{center}
\end{table}

\subsection{Graphs}
The results for the covariates of the riboflavin data were as follows. The procedures {\it thav.glasso} and {\it glasso} were killed after one hour with no results, {\it huge} took 35 seconds but returned zero edges. The Gaussian covariate procedure with the default values took 16 seconds and yielded a directed graph with 4491 edges and an undirected graph with 3882 edges.

\section{Appendix: Technical details and proofs}\label{sec:append}

\subsection{Details and Proofs for Section~\ref{sec:F-test.etc.2}}
\label{app:F-test.etc}

In what follows, we utilize some basic facts about multivariate
Gaussian distributions, see for example \cite{MARKENBIB79}.

\paragraph{Special distributions.}

Let $\b_1,\ldots,\b_p$ be an orthonormal basis of a linear subspace $\V$ of $\R^n$, and let $\Z \sim N_p(\bs{0},\bs{I})$. Then $\tilde{\Z} := \sum_{i=1}^p Z_i \b_i$ has a standard Gaussian distribution on $\V$ with $\|\Z\| = \|\tilde{\Z}\|$. 

The chi-squared distribution with $p$ degrees of freedom coincides with $\text{Gamma}(p/2,2)$, where $\text{Gamma}(a,c)$ stands for the gamma distribution with shape parameter $a > 0$ and scale parameter $c > 0$. The statements of the next Lemma are well known.

\begin{Lemma}
\label{lem:Gamma.Beta}
Let $a,b,c > 0$, and let $Y_a$ and $Y_b$ be independent random variables with distributions $\mathrm{Gamma}(a,c)$ and $\mathrm{Gamma}(b,c)$, respectively. Then $Y_a + Y_b$ and $U := Y_a/(Y_a + Y_b)$ are stochastically independent with $Y_a + Y_b \sim \mathrm{Gamma}(a+b,c)$ and $U \sim \text{Beta}_{a,b}$.
\end{Lemma}
With $Y_a$, $Y_b$ and $U$ as in the previous lemma, $F := (Y_a/a)/(Y_b/b) \sim \text{F}_{2a,2b}$. Note also that $U = (a/b) F/((a/b)F + 1)$ and $1 - U \sim \text{Beta}_{b,a}$. In particular, for $x > 0$,
\begin{align*}
	1 - \text{F}_{2a,2b}(x) \
	&= \ \bs{P}(F \ge x) \ = \ \bs{P} \Bigl( U \ge \frac{(a/b)x}{(a/b)x + 1} \Bigr)
		\ = \ \bs{P} \Bigl( 1 - U \le \frac{1}{(a/b)x + 1} \Bigr) \\
	&= \ \text{Beta}_{b,a} \Bigl( \frac{1}{(a/b)x + 1} \Bigr) .
\end{align*}
With $a = (q- q_0)/2$, $b = (n - q)/2$ and $x = (b/a) (\rss_0 - \rss)/\rss$, we obtain the equation
\[
	1 - \text{F}_{q - q_0, n - q}
		\Bigl( \frac{(\rss_0 - \rss)/(q- q_0)}{\rss/(n - q)} \Bigr)
	\ = \ \text{Beta}_{(n - q)/2, (q- q_0)/2} \Bigl( \frac{\rss}{\rss_0} \Bigr) ,
\]
i.e. equality two of the P-values of  Theorem~\ref{the:PF_PG}.

Lemma~\ref{lem:Gamma.Beta} implies useful facts about products of beta random variables.

\begin{Lemma}
\label{lem:Beta.Beta}
\textbf{(i)} \ For $a,b,c > 0$, let $U \sim \text{Beta}(a,b)$ and $V \sim \text{Beta}(a+b,c)$ be stochastically independent. Then $UV \sim \text{Beta}(a,b+c)$.

\noindent
\textbf{(ii)} \ For $a, \delta > 0$ and $k \in \mathbb{N}$, let $U_1, \ldots, U_k$ be stochastically independent random variables such that $U_j \sim \text{Beta}(a+(j-1)\delta,\delta)$. Then $\prod_{j=1}^k U_j \sim \text{Beta}(a,k\delta)$.
\end{Lemma}

\begin{proof}[\bf Proof of Lemma~\ref{lem:Beta.Beta}]
For proving part~(i), we start with independent random variables $G_a \sim \mathrm{Gamma}(a,1)$, $G_b \sim \mathrm{Gamma}(b,1)$ and $G_c \sim \mathrm{Gamma}(c,1)$. By Lemma~\ref{lem:Gamma.Beta},
\[
	U := \frac{G_a}{G_a + G_b} \sim \text{Beta}(a,b) ,
	\quad
	G_a + G_b \sim \mathrm{Gamma}(a+b,1)
	\quad\text{and}\quad
	G_c
\]
are independent. A second application of Lemma~\ref{lem:Gamma.Beta} implies that the random variables $U$ and
\[	V := \frac{G_a + G_b}{G_a + G_b + G_c} \sim \text{Beta}(a+b,c) 
\]
are also independent so that
\[
	UV \ = \ \frac{G_a}{G_a + G_b + G_c} \ \sim \ \text{Beta}(a,b+c) ,
\]
because $G_a$ and $G_b + G_c \sim \mathrm{Gamma}(b+c,1)$ are independent.

Part~(ii) follows from part~(i) via induction.
\end{proof}

\paragraph{Proof of Theorems~\ref{the:RSS} and \ref{the:PF_PG}}
We consider firstly the case $q_0=q-1$, put 
\[ \mathbb{ V}^{\bot}_0=\{{\bf w} \in \mathbb{R}^n: {\bf
 w}^{\top}{\bf \x}=0 \text{ for all } {\bf \x} \in \mathbb{ V}_0\}\]
where $ \mathbb{ V}_0$ is the linear space spanned by the covariates
${\bf x}_i, i\in {\mathcal M}_0$.

Let ${\bf b}_i, i=1,\ldots,n$ be an orthonormal basis of
$\mathbb{R}^n$ such that
 \[  \mathbb{ V}_0=\text{span}({\bf b}_1,\ldots,{\bf b}_{q_0}) \text{  and }  {\bf b}_{q_0+1}=({\bf y}-P_{
   \mathcal{M}_0}({\bf y}))/(ss_0)^{-1/2}\] 
where $P_{\mathcal{M}_0}$ is the projection onto the subspace $\mathbb{V}_0$. We now replace ${\bf x}_{\nu}$ by a Gaussian covariate ${\bf Z}_{\nu}$  consisting of $n$ i.i.d. $N(0,1)$ random variables. By the rotational symmetry of the standard Gaussian distribution on  $\mathbb{R}^n$, $Z_j:={\bf b}_j^{\top}{\bf Z}_{\nu}$ defines stochastically independent standard Gaussian random variables $Z_1,\ldots,Z_n$. The orthogonal projection of ${\bf Z}_{\nu}$ onto $\mathbb{ V}^{\bot}_0$ is given by 
\[\tilde{{\bf Z}}_{\nu}:=\sum_{j=q_0+1}^nZ_j{\bf b}_j.\]
In particular
\[\text{span}({\bf b}_1,\ldots,{\bf b}_{q_0},{\bf Z})
=\text{span}({\bf b}_1,\ldots,{\bf b}_{q_0},\tilde{{\bf Z}})\] 
and as 
\[P_{\mathcal{M}_1}({\bf y})=P_{\mathcal{M}_0}({\bf
  y})-(ss_0)^{1/2}\frac{\tilde{{\bf Z}}_{\nu}^{\top}{\bf
    b}_{q_0+1}}{\Vert\tilde{{\bf Z}}_{\nu}\Vert^2} \tilde{{\bf Z}}_{\nu}\]  
it follows that
\[SS_1=ss_0-ss_0\frac{(\tilde{{\bf Z}}_{\nu}^{\top}{\bf
    b}_{q_0+1})^2}{\Vert 
  \tilde{{\bf Z}}_{\nu}\Vert^2}\]
and hence
 \begin{equation} \label{equ:dist_RSS}
\frac{SS_1}{ss_0}=1-\frac{(\tilde{{\bf Z}}{\nu}^{\top}{\bf
     b}_{q_0+1})^2}{\Vert 
   \tilde{{\bf Z}}_{\nu}\Vert^2}=\frac{\sum_{j=q_0+2}^n
   Z_j^2}{\sum_{j=q_0+1}^n Z_j^2}\sim
 \text{Beta}((n-q_0-1)/2,1/2).
\end{equation}
 
In the general case with $q-q_0=k>1$ the above argument may be applied inductively to show that
\[\frac{SS_1}{ss_0}=\prod_{\ell=1}^{k}U_{\ell}\]
in distribution where the $U_1,\ldots,U_k$ are stochastically independent with
\[U_{\ell}\sim \text{Beta}((n-q_0-\ell)/2,1/2)\]
We now use the standard result that if $U\sim \text{Beta}(a,b)$ and $V\sim \text{Beta}(a+b,c)$ and $U$ and $V$ are independent then $UV\sim\text{Beta}(a,b+c)$. From this it follows that 
\begin{equation} \label{equ:RSS_dist_gen}
\frac{SS_1}{ss_0}\sim \text{Beta}((n-q)/2,(q-q_0)/2)
\end{equation}
which proves Theorem~\ref{the:RSS} and the first part of Theorem~ \ref{the:PF_PG}.

To prove the second part we note that if $\chi^2_{\nu_1}$ and $\chi^2_{\nu_2}$ are independent chi-squared random variables with $\nu_1$ and $\nu_2$ degrees of freedom respectively then 
\[\frac{\chi^2_{\nu_1}/\nu_1}{\chi^2_{\nu_2}/\nu_2}\sim
\text{F}(\nu_1,\nu_2)\]
and 
\[\frac{\chi^2_{\nu_1}}{\chi^2_{\nu_1}+\chi^2_{\nu_2}}\sim
\text{Beta}(\nu_1/2,\nu_2/2) .\]
From this it follows that for all $x>0$
\[\text{F}_{\nu_1,\nu_2}(x)=\text{Beta}_{\nu_1/2,\nu_2/2}((\nu_1/\nu_2)x/((\nu_1/\nu_2)x+1))\]\cite{DAVDUEM21}
and hence the second equality of the theorem.

\subsection{Details and Proofs for Section~\ref{sec:boun_asymp}}
\label{app:variable.selection}

An important ingredient are bounds for the quantile functions of beta and gamma distributions.

\begin{Lemma}
\label{lem:beta.quantiles}
Let $\text{G}$ be the distribution function of $\mathrm{Gamma}(1/2,2) = \chi_1^2$. Then,
\[
	\text{Beta}_{1/2,(n-1)/2}^{-1} \ \begin{cases}
		\ge \ \text{G}^{-1}/(n-1 + \text{G}^{-1})
			& \text{if} \ n \ge 2 , \\
		\le \ (n - 2)^{-1} \text{G}^{-1}
			& \text{if} \ n \ge 3 .
	\end{cases}
\]
Moreover, for $\delta \in (0,1)$,
\[
	\text{G}^{-1}(1 - \delta)
	\ = \ 2\log(1/\delta) - \log \log(1/\delta) - \log \pi + o(1)
	\quad\text{as} \ \delta \to 0 .
\]
\end{Lemma}

For the second part see for example  Chapter~2 of \cite{DEHFER06}. It has various implications for the maximum of squared standard Gaussian random variables:

\begin{Lemma}
\label{lem:Gumbel}
Let $\Z \in \R^q$ be a random vector with components $Z_i \sim N(0,1)$. Then
\[
	\bs{P} \Bigl( \max_{1 \le i \le q} Z_i^2 \le 2 \log q \Bigr)
	\ \to \ 1
\]
as $q \to \infty$. If $\Z \sim N_q(\bs{0},\bs{I})$, then
\[
	\max_{1 \le i \le q} Z_i^2 \ = \ 2 \log q - \log\log q - \log\pi + 2 X_q
\]
with a random variable $X_q$ such that $\lim_{q\to\infty} \bs{P}(X_q \le x) = \exp(- e^{-x})$ for any $x \in \R$.
\end{Lemma}

Lemma~\ref{lem:beta.quantiles} also leads to a particular approximation of beta quantiles:

\begin{Lemma}
\label{lem:step-wise.quantiles}
For integers $n, q \ge 2$ and fixed $\alpha \in (0,1)$,
\[
	n \text{Beta}_{1/2,(n-1)/2}^{-1} \bigl( (1 - \alpha)^{1/q} \bigr)
	\ = \ 2 \log q - \log\log q - \log \pi - 2 \log(- \log(1 - \alpha)) 
		+ o(1)
\]
as $q \to \infty$ and $n / \log(q)^2 \to \infty$.
\end{Lemma}

\begin{proof}[\bf Proof of Lemma~\ref{lem:beta.quantiles}.]
Recall that $\text{B} := \text{Beta}_{1/2,(n-1)/2}$ is the distribution function of $Z_1^2/(Z_1^2 + S^2)$ with $S^2 = \sum_{i=2}^n Z_i^2$ and $\Z \sim N_n(\bs{0},\bs{I})$. Then Jensen's inequality implies that for $0 < x < 1$,
\[
	\text{B}(x)
	\ = \ \bs{E} \left(\bs{P} \Bigl( Z_1^2 \le \frac{S^2 x}{1 - x} \,\Big|\, S^2 \Bigr)\right)
	\ = \ \bs{E} \left(\text{G} \Bigl( \frac{S^2 x}{1 - x} \Bigr)\right)
	\ \le \ \text{G} \Bigl( \frac{(n-1) x}{1 - x} \Bigr) ,
\]
because $\bs{E}(S^2) = n-1$ and $\text{G}$ is concave. Consequently, for $0 < u < 1$, $\text{B}^{-1}(u)$ is not smaller than the solution $x$ of $(n-1)x/(1 - x) = \text{G}^{-1}(u)$, which is $\text{G}^{-1}(u) / (n - 1 + \text{G}^{-1}(u))$.

On the other hand, if $n \ge 3$, then it it follows from independence of $X := Z_1^2/\|\Z\|^2$ and $T := \|\Z\|^2$ with $\bs{E}(T^{-1}) = (n - 2)^{-1}$ that
\[
	\text{G}(y) \ = \ \bs{P}(TX \le y)
	\ = \ \bs{E} \left(\text{B}(T^{-1} y)\right)
	\ \le \ \text{B}((n-2)^{-1} y)
\]
by Jensen's inequality and concavity of $\text{B}$. Consequently, $\text{B} \ge \text{G}((n-2) \cdot)$, and this implies that $\text{B}^{-1} \ge (n - 2)^{-1} \text{G}^{-1}$.

For the reader's convenience, a proof of the second part is provided as well. Since $\text{G}'(x) = (2\pi x)^{-1/2} e^{-x/2}$, partial integration and elementary bounds yield the inequalities
\[
	2^{1/2} (\pi x)^{-1/2} e^{-x/2} (1 - 2x^{-1})
	\ \le \ 1 - \text{G}(x)
	\ \le \ 2^{1/2} (\pi x)^{-1/2} e^{-x/2} .
\]
If we fix an arbitrary real number $z$ and set $x := 2 \log(1/\delta) - \log\log(1/\delta) - \log\pi + z$, then $x = 2 \log(1/\delta) (1 + o(1)) \to \infty$ and
\[
	2 \log(1 - \text{G}(x)) \ = \ 2 \log(\delta) - z + o(1)
\]
as $\delta \downarrow 0$. This implies the asserted expansion for $\text{G}^{-1}(1 - \delta)$ as $\delta \downarrow 0$.
\end{proof}

\begin{proof}[\bf Proof of Lemma~\ref{lem:step-wise.quantiles}.]
Note first that $(1 - \alpha)^{1/q} = \exp(\log(1 - \alpha)/q)$ may be written as $1 - \delta$ with $\delta := q^{-1} \tilde{\alpha} (1 + O(q^{-1}))$ and $\tilde{\alpha} := - \log(1 - \alpha)$. Since $\log(1/\delta) = \log q - \log \tilde{\alpha} + o(1)$ and $\log \log(1/\delta) = \log(\log q + O(1)) = \log \log q + o(1)$, it follows from the second part of Lemma~\ref{lem:beta.quantiles} that
\[
	\text{G}^{-1} \bigl( (1 - \alpha)^{1/q} \bigr)
	\ = \ 2 \log q - \log \log q - \log \pi - 2 \log \tilde{\alpha} + o(1)
	\ = \ O(\log q)
\]
as $q \to \infty$. Then the first part of that lemma implies that
\begin{align*}
	\text{Beta}_{1/2,(n-1)/2}^{-1} \bigl( (1 - \alpha)^{1/q} \bigr) \
	&= \ (n + O(\log q))^{-1} \text{G}^{-1} \bigl( (1 - \alpha)^{1/q} \bigr) \\
	&= \ n^{-1} \bigl( 1 + O(n^{-1} \log q) \bigr)
		\text{G}^{-1} \bigl( (1 - \alpha)^{1/q} \bigr) \\
	&= \ n^{-1} \bigl( 2 \log q - \log \log q - \log \pi
		- 2 \log \tilde{\alpha} + o(1) \bigr)
\end{align*}
as $q \to \infty$ and $n/\log(q)^2 \to 0$.
\end{proof}

\begin{proof}[\bf Proof of Theorem~\ref{thm:consistency.0}]
Note first that $(\x_\nu^\top\y)^2/\|\y\|^2 = (\x_\nu^\top\Z)^2/\|\Z\|^2$ has distribution function $\text{B} = \text{Beta}_{1/2,(n-1)/2}$. Hence, with $x_{n,q} := \text{B}^{-1} \bigl( (1 - \alpha)^{1/q} \bigr)$,
\[
	\bs{P} \Bigl( \max_{\x_\nu} \frac{(\x_\nu^\top\y)^2}{\|\y\|^2}
		\ge x_{n,q} \Bigr)
	\ \le \ q \bigl( 1 - (1 - \alpha)^{1/q} \bigr)
	\ \le \ - \log(1 - \alpha) ,
\]
because $(1 - \alpha)^{1/q} = \exp \bigl( q^{-1} \log(1 - \alpha) \bigr) \ge 1 + q^{-1} \log(1 - \alpha)$. Note also that $\|\Z\|^2$ has expectation $n$ and variance $2n$, whence for arbitrary $c > 0$,
\[
	\bs{P}(\|\Z\|^2 \le n - cn^{1/2}), \bs{P}(\|\Z\|^2 \ge n + cn^{1/2})
	\ \le \ \frac{2}{2 + c^2}
\]
by the Tshebyshev-Cantelli inequality. Consequently,
\[
	\bs{P} \Bigl( \max_{\x_\nu} \frac{(\x_\nu^\top\y)^2}{\|\y\|^2}
		\ge x_{n,q} \Bigr)
	\ \le \ \bs{P} \Bigl( \max_{\x_\nu} (\x_\nu^\top\Z)^2
			\ge (1 - cn^{-1/2}) n x_{n,q} \Bigr)
		+ \frac{2}{2 + c^2}
\]
and
\[
	\bs{P} \Bigl( \max_{\x_\nu} \frac{(\x_\nu^\top\y)^2}{\|\y\|^2}
		\ge x_{n,q} \Bigr)
	\ \ge \ \bs{P} \Bigl( \max_{\x_\nu} (\x_\nu^\top\Z)^2
			\ge (1 + c^{-1/2}) n x_{n,q} \Bigr)
		- \frac{2}{2 + c^2} .
\]
But it follow from the Gaussian inequality (cf.\ \cite{SID67} or \cite{ROY14}) that for any number $x$,
\[
	\bs{P} \bigl( \max_{\x_\nu} (\x_\nu^\top\Z)^2 \ge x \bigr)
	\ \le \ \bs{P} \bigl( \max_{\nu} Z_\nu^2 \ge x \bigr)
\]
with independent random variables $Z_\nu \sim N(0,1),\nu=1,\ldots,q$ with equality in case of orthonormal regressors $\x_\nu$. Now the claims follow from the fact that for any fixed $c > 0$ and $\tilde{\alpha} := - \log(1 - \alpha)$,
\begin{align*}
	(1 \pm cn^{-1/2}) n x_{n,q} \
	&= \ (1 \pm cn^{-1/2})
		\bigl( 2 \log q - \log\log q - \log\pi - 2\log \tilde{\alpha} + o(1) \bigr) \\
	&= \ 2 \log q - \log\log q - \log\pi - 2\log \tilde{\alpha} + o(1)
\end{align*}
by Lemma~\ref{lem:step-wise.quantiles}, and
\[
	\bs{P} \Bigl( \max_{\nu} Z_\nu^2
		\ge 2 \log q - \log\log q - \log\pi - 2\log\tilde{\alpha} + o(1) \Bigr)
	\ \to \ 1 - \exp(- \exp(\log \tilde{\alpha})) \ = \ \alpha
\]
by Lemma~\ref{lem:Gumbel}.
\end{proof}

\begin{proof}[\bf Proof of Theorems~\ref{thm:consistency.general} and \ref{thm:consistency.orthogonal}]
Note first that in case of orthonormal regressors, $q \le n$, and Condition~(A.1') implies Condition~(A.1). Without loss of generality we assume that $\sigma = 1$.

At first we verify that $\hat{{\mathcal S}} \supset {\mathcal S}_*$ with asymptotic probability one. Having started step-wise selection with ${\mathcal S} = \emptyset$, suppose we have chosen a set ${\mathcal S} \subsetneq {\mathcal S}_*$ of $k$ covariates. The question is whether an additional regressor $\x_\nu$ with $\x_\nu \in {\mathcal S}_* \setminus {\mathcal S}$ will be added to ${\mathcal S}$, regardless of the choice of ${\mathcal S}$. This is certainly the case if
\begin{equation}
\label{ineq:step-wise.i}
	\min_{{\mathcal S} \subsetneq {\mathcal S}_*}
		\Bigl( \max_{\x_\nu \in {\mathcal S}_*\setminus{\mathcal S}} |\x_{{\mathcal S},\nu}^\top\y|
			- \max_{\x_s\notin{\mathcal S}_*} |\x_{{\mathcal S},s}^\top\y| \Bigr)
	\ > \ 0
\end{equation}
and
\begin{equation}
\label{ineq:step-wise.ii}
	\min_{{\mathcal S} \subsetneq {\mathcal S}_*} \Bigl( \max_{\x_\nu \in {\mathcal S}_*\setminus{\mathcal S}} \,
		\frac{|\x_{{\mathcal S},\nu}^\top\y|}{\|Q_{\mathcal S}^{}\y\|}
		- \kappa_{n-k,q-k}^{} \Bigr)
	\ > \ 0
\end{equation}
with asymptotic probability one, where $\kappa_{n',q'} := \sqrt{ B_{1/2,(n'-1)/2}^{-1} \bigl( (1 - \alpha)^{1/q'}\bigr) }$. Inequality~\eqref{ineq:step-wise.i} can be replaced by the stronger but simpler inequality
\begin{equation}
\label{ineq:step-wise.i'}
	\min_{\x_\nu\in{\mathcal S}_*,{\mathcal S} \subset {\mathcal S}_*\setminus\{\x_\nu\},\x_s\notin{\mathcal S}_*}
		\bigl( |\x_{{\mathcal S},\nu}^\top\y|
			- |\x_{{\mathcal S},s}^\top\y| \bigr)
	\ > \ 0 .
\end{equation}
Moreover, according to Lemma~\ref{lem:step-wise.quantiles},
\[
	\max_{0 \le k\le k_*} \, \kappa_{n-k,q-k}
	\ = \ \sqrt{ \frac{(2 + o(1)) \log q}{n} } ,
\]
and $\|Q_{\mathcal S}\y\| \le \|\y\|$. But $\|\y\|^2$ has a non-central chi-squared distribution with $n$ degrees of freedom and non-centrality parameter $\|\bmu\|^2$. In particular, it has expectation $n + \|\bmu\|^2$ and variance $2n + 4\|\bmu\|^2$, and this implies that
\begin{equation}
\label{eq:chi}
	\|\y\| \ = \ \sqrt{n + \|\bmu\|^2} + O_p(1)
	\ = \ \sqrt{n + \|\bmu\|^2} (1 + o_p(1)) .
\end{equation}
Hence we may replace \eqref{ineq:step-wise.ii} with
\begin{equation}
\label{ineq:step-wise.ii'}
	\min_{\x_\nu \in {\mathcal S}_*, {\mathcal S} \subset {\mathcal S}_*\setminus\{\x_\nu\}}
		\frac{|\x_{{\mathcal S},\nu}^\top\y|}{\sqrt{n + \|\bmu\|^2}}
	\ > \ \sqrt{ \frac{\tau' \log q}{n} }
\end{equation}
for some $\tau' > 2$.

Let us verify \eqref{ineq:step-wise.i} and \eqref{ineq:step-wise.ii} for orthonormal regressors $\x_\nu$ and $\bmu = \sum_{\x_\nu\in{\mathcal S}_*} \beta_\nu\x_\nu$. Here $\x_{{\mathcal S},\nu} = \x_\nu$ and ${\mathcal S} \subset {\mathcal S}_*\setminus\{\x_\nu\}$, whence the left hand side of \eqref{ineq:step-wise.i'} equals
\begin{align*}
	\min_{\x_\nu\in{\mathcal S}_*} |\beta_\nu|
		- \max_{\x_\nu\in{\mathcal S}_*} |\x_\nu^\top\Z|
		- \max_{\x_s\notin{\mathcal S}_*} |\x_s^\top\Z| \
	&\ge \ \min_{\x_\nu\in{\mathcal S}_*} |\beta_\nu|
		- \sqrt{2 \log k_*} - \sqrt{2 \log q} - O_p(1) \\
	&\ge \ \sqrt{\tau \log q} - \sqrt{2 \log q} - O_p(1)
		\ \to_p \ \infty ,
\end{align*}
where the second last inequality follows from Lemma~\ref{lem:Gumbel}, and the last inequality is a consequence of Condition~(A.2'). This proves \eqref{ineq:step-wise.i'}. Similarly one can show that the left hand side of \eqref{ineq:step-wise.ii'} is equal to
\begin{align*}
	\min_{\x_\nu\in{\mathcal S}_*} \frac{|\x_\nu^\top\y|}{\sqrt{n + \|\bmu\|^2}} \
	&\ge \ \min_{\x_\nu\in{\mathcal S}_*} \frac{|\beta_\nu| - |\x_\nu^\top\Z|}
		{\sqrt{n + \|\bmu\|^2}} \\
	&\ge \ \min_{\x_\nu\in{\mathcal S}_*} \frac{|\beta_\nu| - \sqrt{2 \log k_*} - O_p(1)}
		{\sqrt{n + \|\bmu\|^2}} \\
	&\ge \ \frac{\sqrt{\tau \log q} - O_p(1)}{\sqrt{n}}
		\ = \ \sqrt{ \frac{(\tau + o_p(1)) \log q}{n} } ,
\end{align*}
and the latter quantity is greater than $\sqrt{\tau' \log(q)/n}$ with asymptotic probability one, provided that $2 < \tau' < \tau$.

Now we verify \eqref{ineq:step-wise.i'} and \eqref{ineq:step-wise.ii'} in the general case. On the one hand, since all vectors $\x_{{\mathcal S},\nu}$ with $\x_\nu \in {\mathcal S}_*$ and ${\mathcal S}\subset{\mathcal S}_*\setminus\{\x_\nu\}$ belong to the unit ball of $\V_{{\mathcal S}_*}$,
\begin{align*}
	\min_{\x_\nu \in {\mathcal S}_*, {\mathcal S} \subset {\mathcal S}_*\setminus\{\x_\nu\}} |\x_{{\mathcal S},\nu}^\top\y| \
	&\ge \ \min_{\x_\nu \in {\mathcal S}_*, {\mathcal S} \subset {\mathcal S}_*\setminus\{\x_\nu\}}
		|\x_{{\mathcal S},\nu}^\top\bmu| - \|\hat{Z}_{{\mathcal S}_*}\| \\
	&\ge \ \min_{\x_\nu \in {\mathcal S}_*, {\mathcal S} \subset {\mathcal S}_*\setminus\{\x_\nu\}}
		|\x_{{\mathcal S},\nu}^\top\bmu| - \sqrt{k_*} - O_p(1) ,
\end{align*}
because $\|\hat{\Z}_{{\mathcal S}_*}\|^2$ has a chi-squared distribution with $m_*$ degrees of freedom, see also the arguments for \eqref{eq:chi}. On the other hand, for any ${\mathcal S} \subset {\mathcal S}_*$ and $\x_s \notin{\mathcal S}_*$, it follows from $\V_{{\mathcal S}}^\perp \supset \V_{{\mathcal S}_*}^\perp$ that the vector $Q_{\mathcal S}\x_s$ is the sum of $Q_{{\mathcal S}_*}\x_s \in \V_{{\mathcal S}_*}^\perp$ and $(Q_{{\mathcal S}} - Q_{{\mathcal S}_*})\x_s \in (\V_{{\mathcal S}_*}^\perp)^\perp = \V_{{\mathcal S}_*}$. Consequently,
\[
	\x_{{\mathcal S},s} \ = \ \lambda_{{\mathcal S},s} \bs{v}_s
		+ \bar{\lambda}_{{\mathcal S},s} \bar{\bs{v}}_{{\mathcal S},s}
\]
with
\begin{align*}
	\lambda_{{\mathcal S},s} \
	&:= \ \|Q_{{\mathcal S}_*}\x_s\| \big/
		\sqrt{ \|Q_{{\mathcal S}_*}\x_s\|^2 + \|(Q_{{\mathcal S}} - Q_{{\mathcal S}_*})\x_s\|^2 }
		\ \in \ [0,1] , \\
	\bs{v}_s \
	&:= \ \|Q_{{\mathcal S}_*}\x_s\|^{-1} Q_{{\mathcal S}_*}\x_s
		\ \in \ \V_{{\mathcal S}_*}^\perp , \\
	\bar{\lambda}_{{\mathcal S},s} \
	&:= \ \sqrt{1 - \lambda_{{\mathcal S},s}^2}
		\ \in \ [0,1] , \\
	\bar{\bs{v}}_{{\mathcal S},s} \
	&:= \ \|(Q_{{\mathcal S}} - Q_{{\mathcal S}_*})\x_s\|^{-1} (Q_{{\mathcal S}} - Q_{{\mathcal S}_*})\x_s
		\ \in \ \V_{{\mathcal S}_*} .
\end{align*}
This implies that
\begin{align*}
	\max_{{\mathcal S} \subset {\mathcal S}_*, \x_s \notin{\mathcal S}_*} |\x_{{\mathcal S},s}^\top\y| \
	&\le \ \max_{{\mathcal S} \subset {\mathcal S}_*, \x_s \notin {\mathcal S}_*}
		|\x_{{\mathcal S},s}^\top\bmu|
		+ \|\hat{\Z}_{{\mathcal S}_*}\|
		+ \max_{s \in \NN\setminus{\mathcal S}_*} |\bs{v}_s^\top\Z| \\
	&\le \ \max_{{\mathcal S} \subset {\mathcal S}_*, \x_s \notin{\mathcal S}_*}
		|\x_{{\mathcal S},s}^\top\bmu|
		+ \sqrt{k_*} + \sqrt{2\log q} + O_p(1) .
\end{align*}
These inequalities and assumption (A.2) imply that the left hand side of \eqref{ineq:step-wise.i'} is not smaller than
\begin{align*}
	\min_{\x_\nu\in{\mathcal S}_*,{\mathcal S}\subset{\mathcal S}_*\setminus\{\x_\nu\},\x_s\notin{\mathcal S}_*}
		& \bigl( |\x_{{\mathcal S},\nu}^\top\bmu| - |\x_{{\mathcal S},s}^\top\bmu| \bigr)
		 - 2 \sqrt{k_*} - \sqrt{2 \log q} - O_p(1) \\
	&\ge \ \sqrt{\tau \log q} - \sqrt{2\log q} - O_p(1)
		\ \to_p \ \infty .
\end{align*}
Hence \eqref{ineq:step-wise.i'} is satisfied with asymptotic probability one. Moreover, a second application of (A.2) shows that the left hand side of \eqref{ineq:step-wise.ii'} is not smaller than
\begin{align*}
	\min_{\x_\nu \in {\mathcal S}_*, {\mathcal S} \subset {\mathcal S}_*\setminus\{\x_\nu\}}
		\frac{|\x_{{\mathcal S},\nu}^\top\bmu| - \|\hat{\Z}_{{\mathcal S}_*}\|}
			{\sqrt{n + \|\bmu\|^2}} \
	&\ge \ \min_{\x_\nu \in {\mathcal S}_*, {\mathcal S} \subset {\mathcal S}_*\setminus\{\x_\nu\}}
		\frac{\sqrt{\tau \log q} + \sqrt{k_*} - O_p(1)}{\sqrt{n}} \\
	&\ge \ \sqrt{ \frac{(\tau + o_p(1)) \log q}{n} } ,
\end{align*}
and the latter quantity is greater than $\sqrt{\tau' \log(q)/n}$ with asymptotic probability one, provided that $2 < \tau' < \tau$.

So far we have shown that with asymptotic probability one, the step-wise selection will lead to the candidate ${\mathcal S} = {\mathcal S}_*$ for $\hat{{\mathcal S}}$. But at that stage, $\x_{{\mathcal S},\nu}^\top\y = \x_{{\mathcal S}_*,\nu}^\top\Z = \x_{{\mathcal S}_*,\nu}^\top Q_{{\mathcal S}_*}\Z$ for all $\x_\nu \notin {\mathcal S}_*$, so
\[
	\bs{P}({\mathcal S}_* \subsetneq \hat{{\mathcal S}})
	\ \le \ o(1) + \bs{P} \Bigl( \max_{\x_\nu \notin {\mathcal S}_*}
		\frac{(\x_{{\mathcal S}_*,\nu}^\top Q_{{\mathcal S}_*}\Z)^2}{\|Q_{{\mathcal S}_*}\Z\|^2}
			\ge \kappa_{n-k_*,q-k_*}^2 \Bigr)
	\ \le \ \alpha + o(1)
\]
by a simple adaptation of Theorem~\ref{thm:consistency.0}.
\end{proof}
%
%
%
%
%

\bibliographystyle{apalike}
\bibliography{literature}
\end{document}